\documentclass[11pt]{article}

\usepackage[all]{xy}           
\usepackage{amsmath,amsfonts,amssymb,amsthm,textcomp}
\usepackage{eucal}
\usepackage{epsfig}

\def\v #1{\vert #1\vert}             

\def\m #1 #2{(-1)^{{\v #1} {\v #2}}} 

\theoremstyle{plain}
\newtheorem{theorem}{Theorem}

\theoremstyle{definition}
\newtheorem{definition}[theorem]{Definition}

\def\<#1>{\langle#1\rangle}

\usepackage{epsfig}
\usepackage{amsmath}
\usepackage{epic}
\usepackage{amssymb}
\usepackage{fancybox}
\usepackage{wrapfig}
\usepackage[font=footnotesize,labelsep=period]{caption}

\usepackage{float}

\usepackage{booktabs}

\usepackage[all]{xy}

\usepackage{fancyhdr}

\usepackage{fancyhdr}

\usepackage{appendix}

\begin{document}

\centerline{\Large \bf A geometric Hamilton--Jacobi theory}\vskip 0.25cm
\centerline{\Large \bf for a Nambu--Poisson structure}

\medskip
\medskip

\centerline{M. de Le\'on and C.
Sard\'on}
\vskip 0.5cm
\centerline{Instituto de Ciencias Matem\'aticas, Campus Cantoblanco}\vskip 0.2cm
\centerline{Consejo Superior de Investigaciones Cient\'ificas}
\vskip 0.2cm
\centerline{C/ Nicol\'as Cabrera, 13--15, 28049, Madrid. SPAIN}

\begin{abstract}


The Hamilton--Jacobi theory is a formulation of Classical Mechanics equivalent to other formulations as Newton's equations, Lagrangian or Hamiltonian Mechanics.
It is particulary useful for the identification of conserved quantities of a mechanical system.
The primordial observation of a geometric Hamilton--Jacobi equation is that if a Hamiltonian vector field $X_{H}$ can be projected into the configuration manifold by means of a 1-form $dW$, then the integral curves of the projected
vector field $X_{H}^{dW}$can be transformed into integral curves of $X_{H}$ provided that $W$ is a solution of the Hamilton--Jacobi equation.
This interpretation has been applied to multiple settings: in nonhonolomic, singular Lagrangian Mechanics and classical field theories.
Our aim is to apply the geometric Hamilton--Jacobi theory to systems endowed with a Nambu--Poisson structure.
The Nambu--Poisson structure has shown its interest in the study physical systems described by several Hamiltonian functions.
In this way, we will apply our theory to two interesting examples in the Physics literature: the third-order Kummer--Schwarz equations and
a system of $n$ copies of a first-order differential Riccati equation. From these examples, we retrieve the original Nambu bracket in three dimensions and
a generalization of the Nambu bracket to $n$ dimensions, respectively.

\end{abstract}

\section{Introduction}

\noindent
The Hamilton--Jacobi equation (HJ equation) is a formulation of Classical Mechanics equivalent to other formulations as Newton's equations, Lagrangian or Hamiltonian Mechanics.
It is particulary useful for the identification of conserved quantities of a mechanical system \cite{Gold,Kibble,LL}.
%
 Let us take the triple $(T^{*}Q,\omega, H)$, where $T^{*}Q$ is our $2n$-dimensional phase space, $\omega$ is a non-degenerate $(0,2)$-skew symmetric tensor
(the canonical symplectic form) and $h$ plays the role of a Hamilton function on $M$. We denote by $(\overline{T^{*}Q}=T^{*}\mathbb{R}\times T^{*}Q,  \overline{\omega}=\omega-dE\wedge dt, \overline{H}=H-E)$ the time-dependent phase space extension of our former set $(T^{*}Q,\omega, H)$. We can choose canonical pairs of coordinates $(q^i,p_j)$ with $i,j=1,\dots,n$ on $T^{*}Q$. 

The standard formulation of the Hamilton--Jacobi theory (HJ theory) consists on finding a function $S(t,q^i)$, called {\it principal function}
such that
\begin{equation}\label{tdepHJ}
 \frac{\partial S}{\partial t}+h\left(q^i,\frac{\partial S}{\partial q^i}\right)=0
\end{equation}
where $h=h(q^i,p_i)$ is the Hamiltonian function of the system. It is possible to absorb the temporal dependency within $S$, with a separation of variables
which directly identifies constants of motion 
\begin{equation}
 S=W(q^1,\dots,q^n)-Et,
\end{equation}
where $E$ is the total energy of the system.
This choice gives rise to the {\it characteristic equation} \cite{Marsden,Gold}
\begin{equation}\label{HJeq1}
 h\left({q}^i,\frac{\partial W}{\partial {q}^i}\right)=E.
\end{equation}

Equations \eqref{tdepHJ} and \eqref{HJeq1} are the so-called $t$-dependent and time independent {\it Hamilton Jacobi equations}, respectively.
The Hamilton--Jacobi equation are a useful intrument to solve the Hamilton equations for $h$
 \begin{equation}\label{hamileq}
 \left\{\begin{aligned}
 {\dot q}^i&=\frac{\partial h}{\partial p_i},\\
 {\dot p}_j&=-\frac{\partial h}{\partial q^j}.
 \end{aligned}\right.
 \end{equation}
Indeed, if we find a solution $W$ of \eqref{HJeq1}, then any solution of \eqref{hamileq} gives a solution of the Hamilton
equations by taking ${p}_i=\partial W/\partial {q}^i.$

The Hamilton--Jacobi theory (HJ theory) has been in the limelight of research during this decade \cite{CGMMMLRR1,CGMMMLRR,LeonIglDiego,MarrSosa,Pauffer,Rund}. In particular, the HJ theory has been widely studied from a geometric point of view.
A lot of geometric results for solving the HJ equation were obtained in \cite{CGMMMLRR1,CGMMMLRR,LeonDiegoMarrSalVil,LeonIglDiego,LeonMarrDiego08,LeonMarrDiegoVaq,MarrSosa}. The primordial observation
for a HJ theory is that if a Hamiltonian
vector field $X_{H}$ can be projected into the configuration manifold by means of a 1-form $dW$, then the integral curves of the projected
vector field $X_{H}^{dW}$can be transformed into integral curves of $X_{H}$ provided that $W$ is a solution of \eqref{HJeq1}.

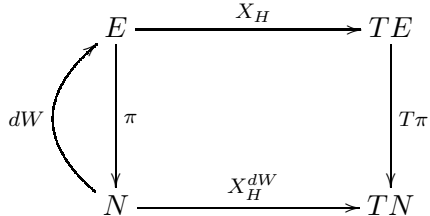
\[
\xymatrix{ E
\ar[dd]^{\pi} \ar[rrr]^{X_H}&   & &TE\ar[dd]^{T\pi}\\
  &  & &\\
 N\ar@/^2pc/[uu]^{dW}\ar[rrr]^{X_H^{dW}}&  & & TN}
\]

\bigskip

Here, $E=T^{*}Q$ for instance. This kind of diagram
has been applied to multiple theories, as nonhonolomic \cite{CGMMMLRR1,CGMMMLRR,LeonIglDiego,LeonMarrDiego}, singular Lagrangian Mechanics \cite{LeonDiegoVaq2,LeonMarrDiegoVaq} and classical field theories \cite{LeonMarrDiego08}
in different geometric settings \cite{CGMMMLRR1,CGMMMLRR,LeonIglDiego,MarrSosa,Pauffer,Rund}.

In particular, the interpretation of the dynamics as lagrangian submanifolds was introduced by W.W. Tulczyjew \cite{Tulczy} who characterized local Hamiltonian
vector fields on a sympletic manifold $(M,\omega)$ as a lagrangian submanifold of the sympletic manifold $(TM, \omega^C)$ where $TM$ is the tangent
bundle of $M$ and $\omega^C$ is the complete tangent lift of $\omega$ to $TM$.
This result was later generalized to Poisson manifolds \cite{GrabUrb} and Jacobi manifolds \cite{IbaLeonMarrDiego1}. In \cite{IbaLeonMarrDiego} it was extended to Nambu--Poisson structures.
Let us recall that Nambu--Poisson structures arose to deal with Hamiltonian systems equipped with two or more Hamiltonian functions. This kind
of system was introduced by Y. Nambu \cite{Nambu}. He considered a bracket of order 3 
\begin{equation*}
 [A,B,C]=\frac{\partial (A,B,C)}{\partial (x,y,z)}
\end{equation*}
with the canonical variables satisfying $[x,y,z]=1$, that could be interpreted as a bracket defined by the canonical volume form in $\mathbb{R}^3$. This bracket
attracted a lot of scientific attention at that time.
An extension to manifolds has been developed by L. Takhtajan in \cite{Takh}. Here, the geometric structure is provided by a contravariant
tensor field $\Lambda$ of order $n$.

The NP structure has shown its interest to study physical systems described in terms of several, compatible Hamiltonian descriptions.
These systems are reputable in the literature of nonlinear phenomena, as it is in Fluid Dynamics, given the recurring existence of an infinite
number of compatible Hamiltonians with a unique system, that are derivable by means of a recursion operator, and that consequently induce multiple conserved quantities \cite{GGKM,Zabusky}.

The plan of the paper is the following:
First, we review the fundamental geometric concepts for (almost) NP structure and fundamental operations. 
We propose a structural theorem, in which the only Lagrangian submanifolds of a NP structure of order $n$ are of dimension $n-1$. This fact motivates the application
of the HJ theory to reduce our system on a general manifold of dimension $n$, to a system defined on an $n-1$-dimensional manifold in which its solution can be retrieved in an easier manner.
By relying on the Darboux integrability theorem elucidated for the global NP structure in \cite{IbaLeonMarrDiego}.
 We derive an expression for the Hamilton--Jacobi equation adapted to $n$-dimensional volume NP structures.
Afterwards, we illustrate the proposed theory by a list of examples. In particular, we study the third-order Kummer--Schwarz equation, which retrieves
the former Nambu structure proposed by Nambu for dimension 3, then we generalize to an example in which an $n$-dimensional bracket is present. This is the case
of $n$-coupled first-order Riccati differential equations. For both $3$ and $n$ dimensional examples, we derive their associated volume form.
To finish, the HJ theory is discussed for the proposed examples.

Along the paper, we consider all the manifolds to be $C^{\infty}$ and we will denote by $\mathfrak{X}(M)$ the Lie algebra of vector fields defined
on $M$, $C^{\infty}(M)$ is the Lie algebra of $C^{\infty}$ functions and $\Lambda^k(T^{*}M)$ are the differential forms of order $k$ on M.

\section{Nambu--Poisson manifolds}

Let us consider an almost Nambu--Poisson manifold, i.e., the pair $(E,\Lambda)$ where $E$ is a differentiable manifold of dimension $m$ equipped with a $(n,0)$-skew symmetric contravariant tensor $\Lambda$ $(m\geqslant n)$. 
The tensor $\Lambda$ defines the vector bundle morphism $\sharp:\Lambda^{n-1}(T^{*}E)\rightarrow TE$ by $<\sharp(\alpha),\beta>=\Lambda(\alpha,\beta)$,
where $\alpha\in \Lambda^{n-1}(T^{*}E)$ and $\beta \in T^{*}E$.

The bracket induced by $\Lambda$ on $C^{\infty}(E)$ is defined as
\begin{equation}
\{f_1,\dots,f_n\}=\Lambda(df_1,\dots,df_n),\qquad f_1,\dots,f_n\in C^{\infty}(E)
\end{equation}
This bracket has the following properties
\begin{enumerate}
 \item $\{f_1,\dots,f_n\}=(-1)^{\epsilon(\sigma)}\{f_{\sigma(1)},\dots,f_{\sigma(n)}\},$
with $\sigma \in \text{Symm}(n)$ and $\epsilon(\sigma)$ is the parity of the permutation;
\item $\{f_1g_1,\dots,f_n\}=f_1\{g_1,\dots,f_n\}+g_1\{f_1,\dots,f_n\},$
\end{enumerate}
which are the skew-symmetry and Leibnitz rule, correspondingly.

We call $\mathcal{D}=\sharp \Lambda^{n-1}(T^{*}E)$ the characteristic distribution induced by $\Lambda$ 
where $\mathcal{D}_x=\sharp \Lambda^{n-1}(T^{*}_xE)$ for every point $x\in E$. Notice that
this is a generalized distribution, in the sense that $r^X:x\in E\mapsto \text{dim} \mathcal{D}_x$ is not necessarily constant on $E$.
We refer to the annihilator of the distribution $\mathcal{D}$ by $\mathcal{D}^{o}=\text{ker}(\sharp)$.


In the case $n=2$ for $\Lambda$ tensors, we retrieve the ordinary almost-Poisson tensors with the integrability condition of the Schouten-Nijinheus bracket $[\Lambda,\Lambda]=0$.
But in general, for an arbitrary order of $\Lambda$, the integrability condition reads $[\Lambda,\Lambda]=(-1)^{n^2}[\Lambda,\Lambda]$, 
which is trivially satisfied for an odd $n\geqslant 3$. In this way, we are in need of another characterization of integrability for manifolds equipped with $(n,0)$-skew symmetric tensors. 
We introduce a generalization of the Jacobi identity, discussed by Takhtajan for $n$-dimensional brackets and dynamics \cite{Takh}.

Consider now an almost Nambu--Poisson manifold $(E,\Lambda)$ with $m\geq n\geq 3$.
To have dynamics, we are provided with $C^{\infty}(E)$ hamiltonian functions $f_1,\dots,f_{n-1}:E\rightarrow \mathbb{R}$ whose corresponding vector field is  
the Hamiltonian vector field
$$X_{f_1,\dots,f_{n-1}}=\sharp(df_1\wedge\dots\wedge df_{n-1}).$$
When all these vector fields are derivations of the algebra $C^{\infty}(E)\times \dots \times C^{\infty}(E)$,
that is, the following identity, known as {\it fundamental identity} introduced by Takhtajan \cite{Takh}, holds
\begin{equation}\label{fundidentity}
 X_{f_1,\dots,f_{n-1}}\{g_1,\dots,g_n\}=\sum_{i=1}^n \{g_1,\dots,X_{f_1,\dots,f_{n-1}}g_i,\dots,g_n\}
\end{equation}
for all functions $f_1,\dots,f_{n-1},g_1,\dots,g_n\in C^{\infty}(E)$ on $E$, then, $(E,\Lambda)$ is called a {\it Nambu--Poisson manifold}
and $\Lambda$ is a {\it Nambu--Poisson tensor}.

We now introduce a theorem depicted in \cite{IbaLeonMarrDiego} which is key to the understanding of the forthcoming sections.

\begin{theorem}
 Let $(E,\Lambda)$ be a generalized $m$-dimensional almost Poisson manifold of order $n\geqslant 3$. 
\begin{itemize}
 \item If $\Lambda$ is a generalized Poisson tensor, the distribution $\mathcal{D}$ is not involutive in general.
\item If $\Lambda$ is a Nambu--Poisson tensor, then the distribution $\mathcal{D}$ is completely integrable and defines a foliation on $E$ such that
when $\Lambda$ is restricted to leaves of the foliation, there exist induced Nambu--Poisson structures in each leaf. The leaves are of two kinds,
for a point $x\in E$, if $\Lambda(x)\neq 0$, then the leave passing through $x$ has dimension $n$ and the induced Nambu--Poisson structure derives
from a volume form. In this way, we can choose local coordinates $\{x^1,\dots,x^n,x^{n+1},\dots,x^m\}$ in which the tensor can be locally
written as
\begin{equation}
 \Lambda=\frac{\partial}{\partial x^1}\wedge \dots \wedge \frac{\partial}{\partial x^n}.
\end{equation}
Associated with this tensor, there exists a volume form which can be locally written as
\begin{equation}
 \Omega=dx^1\wedge \dots \wedge dx^n.
\end{equation}
\noindent
If $\Lambda=0$, then the leaf reduces to a point $x$ and the induced Nambu--Poisson structure is trivial.
\end{itemize}

\end{theorem}

\begin{proof}
 Complete proof of this theorem can be found in one of the articles of the present authors \cite{IbaLeonMarrDiego}.
\end{proof}
This theorem is key to the following sections.
\subsection*{Lagrangian submanifolds}

Let $(E,\Lambda)$ be a Nambu--Poisson manifold with $m>n\geq 3$, we say that a submanifold $N\subset E$ is $j$-Lagrangian $\forall x\in N, 1\leqslant j\leqslant n-1$ if
\begin{equation}\label{lagsubm}
 \sharp \text{Ann}^j(T_{x}N)=\sharp (\Lambda^{n-1} (T_{x}^{*}E))\cap T_{x}N,
\end{equation}
where the annihilator is defined as
\begin{equation}\label{lagsubmann}
 \text{Ann}^j(T_{x}N)=\{\alpha\in \Lambda^{n-1}(T^{*}_xE)|\quad \iota_{v_1\wedge \dots\wedge v_j} \alpha=0,\forall v_1,\dots,v_j \in T_{x}N\}.
\end{equation}
The following inclusions are clearly fulfilled
\begin{equation}
 \text{Ann}^{1}(T_{x}N)\subseteq \text{Ann}^{2}(T_{x}N)\subseteq \dots \subseteq \text{Ann}^{n-1}(T_{x}N).
\end{equation}

\subsection*{Nambu--Poisson volume manifolds}

\noindent
We consider a volume manifold as a pair $(E,\Omega)$, where $\Omega$ is a volume form on the differentiable $n$ dimensional manifold $E$.
There is an associated $(n,0)$-skew symmetric tensor $\Lambda_{\Omega}$ defined as
\begin{equation}\label{need1}
\Lambda_{\Omega}(df_1 \dots df_n)=\{f_1,\dots,f_n\}
\end{equation}
where the bracket is defined by
\begin{equation}\label{need2}
 \{f_1,\dots,f_n\}\Omega=df_1\wedge \dots \wedge df_n.
\end{equation}

A particular example is $E\simeq \mathbb{R}^n$ with canonical coordinates $\{x^i,i=1,\dots,n\}$. Here, the canonical volume form is written as $\Omega_{\mathbb{R}^n}=dx^1\wedge \dots \wedge dx^n$
and the bracket reduces to the Jacobian
\begin{equation}\label{defjacobian}
 \{f_1,\dots,f_n\}=\left|\begin{array}{ccc}
                          \frac{\partial f_1}{\partial x^1}& \dots& \frac{\partial f_1}{\partial x^n}\\
                             \frac{\partial f_2}{\partial x^1}& \dots& \frac{\partial f_2}{\partial x^n}\\
                            \dots& \dots& \dots \\
                          \frac{\partial f_n}{\partial x^1}& \dots& \frac{\partial f_n}{\partial x^n}
                         \end{array}\right|
\end{equation}
which for the case $n=3$, we retrieve the original bracket introduced by Nambu \cite{Nambu2,Nambu}.

\begin{theorem}\label{th1}
 Given a volume Nambu--Poisson structure $(E,\Omega)$ of dimension $n$, every submanifold of codimension $1$
is $(n-1)$-Lagrangian. No other Lagrangian submanifolds exist.

\end{theorem}
\begin{proof}

Choose local coordinates $(x^1,\dots,x^n)$ such that $\Omega=dx^1\wedge \dots \wedge dx^n$. A direct computation shows that
\begin{equation*}
 \sharp \left(dx^1\wedge \dots \wedge d\check{x^{i}}\wedge \dots \wedge dx^n\right)=(-1)^{n-i}\frac{\partial}{\partial x^i}
\end{equation*}
where $d\check{x^{i}}$ stands for the omitted term $dx^{i}.$
This implies that the generalized distribution is locally spanned by $\{\frac{\partial}{\partial x^1},\dots,\frac{\partial }{\partial x^n}\}$,
or in other words,
\begin{equation*}
 \sharp (\Lambda^{n-1} T^{*}E)=TE.
\end{equation*}
Now, assume that $N$ is a $(n-1)$-dimensional submanifold of $E$ such that in local coordinates, it is defined as $\{x^n=0\}$.
Therefore, $TN$ is locally generated by $\{\frac{\partial}{\partial x^1},\dots,\frac{\partial }{\partial x^{n-1}}\}.$
Take $\alpha \in \Lambda^{n-1}T^{*}E$; in local coordinates
\begin{equation*}
 \alpha=\alpha_i \; dx^1\wedge \dots \wedge d\check{x^{i}}\wedge \dots \wedge dx^n
\end{equation*}
where as above, the ``check'' symbol over a term means that this term is omitted.
We compute 
$$\text{Ann}^j(T^{*}N)$$
for a fixed $j$, $1\leq j\leq n-2$.

Since $\alpha \in \text{Ann}^j(T^{*}N)$ if and only if $i_{v_1\wedge \dots \wedge v_j}\alpha=0,$ we deduce that
\begin{equation*}
 i_{\frac{\partial}{\partial x^{l_1}}\wedge \dots \wedge \frac{\partial}{\partial x^{l_j}}}\alpha=0
\end{equation*}
for $1\leq l_1< \dots <l_j\leq n-1$.
Then,
$$\text{Ann}^{j}T^{*}N=\{0\}.$$
Indeed, since
\begin{equation*}
 i_{\frac{\partial}{\partial x^{l_1}}\wedge \dots \wedge \frac{\partial}{\partial x^{l_j}}}\alpha^i dx^1\wedge \dots \wedge d\check{x^{i}}\wedge \dots dx^n=0, 
\end{equation*}
for all $1\leq l_1< \dots < l_j\leq n-1,$ it implies that $\alpha_i=0$ for all $i$.

Now, we compute,
$$\text{Ann}^{n-1}(T^{*}N)$$
In this case, $\alpha \in \text{Ann}^{n-1}(T^{*}N)$ if and only if
\begin{equation*}
 i_{\frac{\partial}{\partial x^{l_1}}\wedge \dots \wedge \frac{\partial}{\partial x^{l-1}}}\alpha=0,\quad \forall 1\leq l_1< \dots<l_{n-1}\leq n-1
\end{equation*}
and therefore, a direct computation shows that
\begin{equation*}
\text{Ann}^{n-1}(T^{*}N)=<\alpha_i dx^1\wedge \dots \wedge d\check{x^{i}}\wedge \dots \wedge dx^{n}>
\end{equation*}
with $1\leq i\leq n-1.$
And thus,
\begin{equation*}
 \sharp \left(\text{Ann}^{n-1} T^{*}N\right)=\langle \frac{\partial}{\partial x^1},\dots,\frac{\partial}{\partial x^{n-1}}\rangle
\end{equation*}
Consequently,
\begin{equation*}
 \sharp \left(\text{Ann}^{n-1} T^{*}N\right)=TN
\end{equation*}
and $N$ is $(n-1)$-lagrangian.
For submanifolds $N$ of codimension greater than $1$, one can use similar arguments to show that
\begin{equation*}
 \text{Ann}^{n-1} T^{*}N=\{0\}.
\end{equation*}
Therefore, they cannot be lagrangian submanifolds.

\end{proof}



\section{Hamilton--Jacobi theory on Nambu--Poisson manifolds}

\noindent
Given a Nambu--Poisson structure $(E,\Lambda)$, consider the map $\sharp: \Lambda^{n-1}(E)\rightarrow \mathfrak{X}(E)$ induced by $\Lambda$. Let us choose a set of functions $f_1,\dots,f_{n-1}$ in $C^{\infty}(E)$
and define the pairing 
$$<\sharp (df_1\wedge \dots \wedge df_{n-1}),df_n>=\Lambda(df_1,\dots,df_{n-1},df_n),$$
where $df_1,\dots,df_{n-1},df_n \in \Omega^{1}(E)$ are one-forms in $E$. 
The characteristic distribution in this case is $\mathcal{D}_x=\sharp \Lambda^{n-1}(T_{x}^{*}E)$ and the associated Hamiltonian vector field 
is defined by
$$X_{f_1,\dots,f_{n-1}}=\sharp (df_1\wedge \dots \wedge df_{n-1}).$$
In particular, we are interested in scenarios with a volume Nambu--Poisson structure $(E,\Omega)$ with $\textit{dim}\ E=n$, whose
dynamics is interpreted in terms of $(n-1)$-Hamiltonian functions $H_1,\dots,H_{n-1}\in C^{\infty}(E)$, in which the Hamilton--Jacobi theory is applicable.

Here, we assume that the $n$-dimensional manifold $E$ fibers over a manifold $N$ of dimension $n-1$, say $\pi:E\rightarrow N$ is a fibration.
Given a section $\gamma$ of $\pi$, that is, $\gamma : N \longrightarrow E$ is such that $\pi \circ \gamma = Id_N$, then $\gamma(N)$ is a submanifold of $E$ with codimension $1$.
\noindent
The vector field $X_{H_1,\dots H_{n-1}}^{\gamma}$ is then defined as 
\begin{equation}\label{defvf}
X_{H_1,\dots H_{n-1}}^{\gamma}=T\pi\circ X_{H_1,\dots H_{n-1}} \circ \gamma
\end{equation}
The following diagram summarizes the above construction
\[
\xymatrix{ (E, \Omega)
\ar[dd]^{\pi} \ar[rrr]^{X_{H_1,\dots H_{n-1}}}&   & &TE\ar[dd]^{T\pi}\\
  &  & &\\
 N\ar@/^2pc/[uu]^{\gamma}\ar[rrr]^{X_{H_1,\dots H_{n-1}}^{\gamma}}&  & & TN }
\]

\begin{theorem}
The vector fields $X_{H_1,\dots,H_{n-1}}$ and $X_{H_1,\dots,H_{n-1}}^{\gamma}$ are $\gamma$-related if and only if the following equation
is satisfied
\begin{equation}
 d(H_1\circ \gamma)\wedge \dots \wedge d(H_{n-1}\circ \gamma)=0.
\end{equation}

\end{theorem}

\begin{proof}
 According to definition  
$$X_{H_1,\dots,H_{n-1}}=\sharp (dH_1\wedge \dots \wedge dH_{n-1}),$$ 
being $\Lambda=\frac{\partial}{\partial x^1}\wedge \dots \wedge \frac{\partial}{\partial x^n},$ we have
\begin{align*}
X_{H_1,\dots,H_{n-1}}=\sharp &\left(\sum_{k=1}^{n-1}\frac{\partial (H_1,\dots, H_{n-1})}{\partial (x^1,\dots \check{x^k},\dots x^n)}dx^1\wedge \dots \wedge d\check{x^{k}}\wedge \dots \wedge dx^n\right)\\
&=\sum_{k=1}^{n-1}(-1)^{n-k}\frac{\partial (H_1,\dots, H_{n-1})}{\partial (x^1,\dots \check{x^k},\dots x^n)}\frac{\partial}{\partial x^k}+\frac{\partial (H_1,\dots, H_{n-1})}{\partial (x^1,\dots x^{n-1})}\frac{\partial}{\partial x^n}
\end{align*}
such that $ X_{H_1,\dots,H_{n-1}}^{\gamma}$ is the projection to $TN$
\begin{equation*}
 X_{H_1,\dots,H_{n-1}}^{\gamma}=\sum_{k=1}^{n-1}(-1)^{n-k}\frac{\partial (H_1,\dots, H_{n-1})}{\partial (x^1,\dots \check{x^k},\dots x^n)}\frac{\partial}{\partial x^k},\qquad k=1,\dots,n-1.
\end{equation*}
Assume that $\gamma$ has the local expression in fibered coordinates 
$$\gamma=(x^1,\dots,\gamma^{n}(x^1,\dots,x^{n-1}))$$
Therefore, we have
\begin{align*}
\text{T}_{\gamma} X_{H_1,\dots,H_{n-1}}^{\gamma}=\sum_{k=1}^{n-1}\frac{\partial (H_1,\dots,H_{n-1})}{\partial (x^1, \dots,\check{x^k},\dots, x^n)}\text{T}_{\gamma}\left(\frac{\partial}{\partial x^k}\right) \\
=\sum_{k=1}^{n-1}\frac{\partial (H_1,\dots,H_{n-1})}{\partial (x^1, \dots,\check{x^k},\dots, x^n)} \left(\frac{\partial}{\partial x^k}+\frac{\partial \gamma^n}{\partial x^k}\frac{\partial}{\partial x^n}\right)
\end{align*}
Then $X_{H_1,\dots,H_{n-1}}$ and $X_{H_1,\dots,H_{n-1}}^{\gamma}$ are $\gamma$-related if and only if
\begin{equation}\label{echamj}
 \sum_{k= 1}^{n}\frac{\partial (H_1,\dots,H_{n-1})}{\partial (x^1, \dots,\check{x^k},\dots, x^n)}\frac{\partial \gamma^n}{\partial x^k}=0,
\end{equation}
which is equivalent to
\begin{equation}\label{HJeq}
 d(H_1\circ \gamma)\wedge \dots \wedge d(H_{n-1}\circ \gamma)=0.
\end{equation}
\end{proof}
The last expression \eqref{HJeq} receives the name of {\it Hamilton--Jacobi equation on a volume Nambu--Poisson manifold}. We say that $\gamma$ is a {\it solution of
the Hamilton--Jacobi problem on a volume Nambu--Poisson manifold} for $(E,\Lambda,\Omega,H_1,\dots,H_{n-1})$ of degree and dimension $n$.

Next, we consider a general Nambu--Poisson manifold $(E,\Lambda)$ of order $n$ and dimension $m$, such that $\pi:E\rightarrow N$ is a fibration
over an $n$-dimensional manifold $N$. Take $\gamma:N\rightarrow E$ a section of $\pi$ such that $\pi\circ \gamma=\text{Id}_N$ and  $\gamma(N)$
is a Lagrangian submanifold of $(E,\Lambda)$. An additional hypothesis is that $\gamma(N)$ has a clean intersection with the leaves of the characteristic
foliation $\mathcal{D}$ defined by $\Lambda$. We recall that this implies that for each leaf $L\in \mathcal{D}$,
\begin{enumerate}
 \item $\gamma(N)\cap L$ is a submanifold
\item $\text{T}(\gamma(N)\cap L)=\text{T}\gamma(N)\cap \text{T}N$
\end{enumerate}
If we are assuming that $\gamma(N)$ is a $j$-Lagrangian submanifold of $(E,\Lambda)$, then $\gamma(N)\cap L$ is a $j$-Lagrangian submanifold
of $L$ with the restricted Nambu--Poisson structure, that according to Theorem 2.1, is a volume structure.
Consequently, $j=n-1$ and $N$ has dimension $n-1$.
Now, let $H_1,\dots,H_{n-1}$ be Hamiltonian functions in $E$ and $X_{H_1 \dots H_{n-1}}$ the corresponding Hamiltonian vector field. 
We define the vector field on $N$,
\begin{equation*}
 X_{H_1\dots H_{n-1}}^{\gamma}=\text{T}_{\pi}\circ X_{H_1 \dots H_{n-1}}\circ \gamma.
\end{equation*}
Since every Hamiltonian vector field is tangent to the characteristic foliation, one can conclude that
\begin{theorem}
 The vector fields $X_{H_1\dots H_{n-1}}^{\gamma}$ and $X_{H_1 \dots H_{n-1}}$ are $\gamma$-related if and only if
\begin{equation}\label{HJeq3}
 d(H_1\circ \gamma)\wedge \dots \wedge d(H_{n-1}\circ \gamma)=0
\end{equation}

\end{theorem}
Therefore, \eqref{HJeq3} will be called the {\it HJ equation for a general Nambu--Poisson manifold} and $\gamma$ satisfying \eqref{HJeq3}
will be a {\it solution of the HJ problem on a general Nambu--Poisson manifold} for $(E,\Lambda,H_1,\dots,H_{n-1})$ of degree $n$ and dimension $m$. 

\begin{definition}
 Given the fibration $\pi:E\rightarrow N$, we say that the bijective map $\Phi:N\times \mathbb{R}\rightarrow E$ is a {\it complete solution of the Hamilton--Jacobi problem}
if
\begin{enumerate}
\item $(\pi\circ \Phi)(x,\lambda)=x$, i.e., $\pi\circ \Phi=\text{pr}_N$.
 \item $\Phi$ is a local diffeomorphism.
\item For all $\lambda \in \mathbb{R}$, we construct a uniparametric family $\Phi_{\lambda}$ such that
\begin{equation*}
\begin{array}{rccl}
\Phi_{\lambda}: &  N &\longrightarrow &  E \\ \noalign{\medskip}
&x& \rightarrow & \Phi_{\lambda}(x)=\Phi(x,\lambda)
\end{array}
\end{equation*}

is a solution of the Hamilton--Jacobi problem.
\end{enumerate}
\end{definition}
\noindent
It can be interpreted geometrically in terms of the following diagram

\[
\xymatrix{
& N\times \mathbb{R} \ar[dd]^{\Phi} \ar[dl]_{\pi\circ \Phi} \\
N \\
&E\ar[ul]^{\pi}
}
\]
Notice that $\Phi_{\lambda}$ is a section.
We now define a function $f=\text{pr}_{\mathbb{R}}\circ \Phi^{-1}$

\[
\xymatrix{
& N\times \mathbb{R} \ar[dd]^{\Phi} \ar[dl]_{\text{pr}_{N}} \ar[dr]^{\text{pr}_{\mathbb{R}}} \\
N & &\mathbb{R}\\
&E\ar[ul]^{\pi} \ar[ur]_{f\circ \Phi_{\lambda}=\text{const.}} 
}
\]
such that $f\circ \Phi_{\lambda}$ is a constant $\left((f\circ \Phi_{\lambda}(x)=f(x,\lambda)=\lambda\right)$
and $d(f\circ \Phi_{\lambda})=0$.

\section{Applications}

\noindent
We consider classical mechanical systems described in terms of  several Hamiltonian functions $H_1,\dots,H_{n-1}$ on the cotangent bundle
$T^{*}\mathbb{R}^n$ of the configuration manifold $E\simeq \mathbb{R}^n$. 
We take canonical bundle coordinates $\{x^1,\dots,x^n\}$ such that
the canonical projection $\pi_{\mathbb{R}^n}(x^1,\dots,x^n)=(x^1,\dots,x^{n-1})$.
In bundle coordinates,
\begin{equation}\label{bivec}
 \Lambda_{\Omega_{\mathbb{R}^n}}=\frac{\partial}{\partial x^1}\wedge \dots \wedge \frac{\partial}{\partial x^n}
\end{equation}
and 
\begin{equation}
 \Omega_{\mathbb{R}^n}=dx^1\wedge \dots \wedge dx^n
\end{equation}
The section $\gamma$ will locally be expressed as 
$$(x^1,\dots,x^{n-1},\gamma^n(x^1,\dots,x^{n-1}))$$
and the projected Hamiltonian vector field is given by
\begin{equation}
X_{H_1,\dots,H_{n-1}}^{\gamma}=\sum_{k=1}^{n-1} \left(\frac{\partial (H_1,\dots,H_{n-1})}{\partial (x^1,\dots,\check{x^k},\dots,x^n)}\circ \gamma\right)\frac{\partial}{\partial x^k}.
\end{equation}

\subsection{The Kummer--Schwarz equation}

\noindent
Let us consider a straightforward application of the former Nambu--Poisson bracket introduced by Nambu \cite{Nambu2,Nambu}, which corresponds with
the case $n=3$.
Consider a third-order
Kummer--Schwarz equation (3KS equation) \cite{Berk,GL} of the form
\begin{equation}\label{KS3}
\frac{d^3x}{dt^3}=\frac 32\left(
\frac{dx}{dt}\right)^{-1}\!\!\left(\frac{d^2x}{dt^2}\right)^{2}\!\!-2c_0\left(\frac{
dx}{dt}\right)^3\!\!+2b_1(t)\frac{dx}{dt},
\end{equation}
where $c_0$ is a real constant and $b_1=b_1(t)$ is any
$t$-dependent function. The physical interest of this equation resides in their relation to the Kummer's problem \cite{Berk} and the Milne--Pinney
and Riccati equations \cite{Ioffe}. In the case in which $c_0=0$, we retrieve the Schwarzian derivative of a function $x(t)$ with respect to $t$ \cite{Leach}.
The 3KS also makes appearances in relation to the $t$-harmonic oscillators in two body problems \cite{Bekov} and Quantum Mechanics \cite{Ioffe}.
Lower dimensional reductions of the 3KS are particular cases of the Gambier equations and similar cosmological models. 

The 3KS equation can be rewritten as a system of first-order differential equations obtained by adding the variables $v\equiv dx/dt$ and $a\equiv d^2x/dt^2$, namely
\begin{equation}\label{firstKS3}
\frac{dx}{dt}=v,\qquad \frac{dv}{dt}=a,\qquad \frac{da}{dt}=\frac 32 \frac{a^2}v-2c_0v^3+2b_1(t)v\,,
\end{equation}	
Its associated to the $t$-dependent vector field reads
\begin{equation}\label{Ex}
X^{3KS}_t=v\frac{\partial}{\partial x}+a\frac{\partial}{\partial v}+\left(\frac 32
\frac{a^2}v-2c_0v^3+2b_1(t)v\right)\frac{\partial}{\partial
a}=Y_3+b_1(t)Y_1,
\end{equation}
where the vector fields on $\mathcal{O}_2=\{(x,v,a)\in{T}^2\mathbb{R}\mid v\neq 0\}$ given by
\begin{align}\label{VFKS1}
Y_1=&2v\frac{\partial}{\partial a},\quad Y_2=v\frac{\partial}{\partial v}+2a\frac{\partial}{\partial a},\nonumber\\
&Y_3=v\frac{\partial}{\partial x}+a\frac{\partial}{\partial v}+\left(\frac 32\frac{a^2}v-2c_0v^3\right)\frac{\partial}{\partial a},
\end{align}
 satisfy
the commutation relations
$$[Y_1,Y_3]=2Y_2,\quad [Y_1,Y_2]=Y_1,\quad [Y_2,Y_3]=Y_3.$$
These vector fields span a three-dimensional Lie algebra of vector fields $V^{3KS}$ isomorphic
to $\mathfrak{sl}(2,\mathbb{R})$ and $X^{3KS}$ becomes a $t$-dependent vector field taking values in $V^{3KS}$. We say that $X^{3KS}$ is a {\it  Lie system},
according to the general theory of Lie systems \cite{CGLS,CLS1,CLS122,LS,Sardon}. The theory of Lie systems has been in vogue of research during the last decades,
achieving beautiful geometric properties of many reputable nonlinear dynamical systems spread in the scientific literature. For example, Lie systems
guarantee the existence of general solutions in terms of nonlinear superposition rules.
The theory has been widely discussed under different geometric frameworks: symplectic, presymplectic, Poisson setting among others.
In particular, it has been studied from the point of view of an underlying Dirac structure \cite{CGLS}. As a particular case of this Dirac structure, 
we can endow the
manifold $\mathcal{O}_2$ with a presymplectic form $\omega_{3KS}$ in such a way that $V^{X^{3KS}}$ consists of Hamiltonian vector fields with respect to it. Indeed, by considering the equations $
\mathcal{L}_{Y_1}\omega_{3KS}=\mathcal{L}_{Y_2}\omega_{3KS}=\mathcal{L}_{Y_3}\omega_{3KS}=0$ and $d\omega_{3KS}=0$, we can readily find the presymplectic form \cite{CGLS}
\[
\omega_{3KS}=\frac{dv\wedge da}{v^3}
\]
on $\mathcal{O}_2$. Additionally, we see that

\begin{equation}\label{3KSHamFun}
\iota_{Y_1}\omega_{3KS}=d\left(\frac{2}{v}\right),\quad \iota_{Y_2}\omega_{3KS}=d\left(\frac{a}{v^2}\right), \quad \iota_{Y_3}\omega_{3KS}=d\left(\frac{a^2}{2v^3}+2c_0v\right).
\end{equation}

From (\ref{3KSHamFun}), it follows that the vector fields $Y_1$, $Y_2,$ and $Y_3$ have Hamiltonian
functions
\begin{equation}\label{Fun3KS}
h_1=-\frac{2}{v},\qquad h_2=-\frac{a}{v^2},\qquad h_3=-\frac{a^2}{2v^3}-2c_0v,
\end{equation}
respectively.
Moreover, 
$$\{ h_1, h_3 \}= 2 h_2, \quad \{h_1,h_2 \} = h_1, \quad \{ h_2, h_3 \}=h_3,$$
where $\{\cdot,\cdot\}$ is the Poisson bracket on $(\mathcal{O}_2,\omega_{3KS})$  induced by $\omega_{3KS}$. In consequence, $h_1,h_2, $ and $h_3$ span a finite-dimensional real Lie algebra
isomorphic to $\mathfrak{sl}(2,\mathbb{R})$. According to the theory of Lie--Hamilton systems \cite{Sardon}, there exists a t-dependent Hamiltonian function 
\begin{equation}\label{ham1}
h^{3KS}_t = h_3 + b_1(t)h_1.
\end{equation}
Furthermore, these vector fields are Hamiltonian with respect to a second presymplectic form \cite{CGLS}, which 
is obtained in the following way
\[
\omega_{Z_P}=-\frac{2}{v^3}(xdv\wedge da+vda\wedge dx+adx\wedge dv)
\]
by $\mathcal{L}_{Z_P}\omega_{3KS}$, where $Z_P=x^2\partial/\partial x$ is a symmetry of of the 3KS equation.
We obtain another triple of Hamiltonian functions with respect to this presymplectic form
\[
\begin{gathered}
\iota_{Y_1}\omega_{Z_P}=-d({{Z_P}}h_1)=-d\left(\frac{4x}{v}\right),\quad
\iota_{Y_2}\omega_{Z_P}=-d({{Z_P}}h_2)=d\left(2-\frac{2ax}{v^2}\right),\\
\iota_{Y_3}\omega_{Z_P}=-d({{Z_P}}h_3)=d\left(\frac {2a}v-\frac{a^2x}{v^3}\right).
\end{gathered}
\]
So, $Y_1,Y_2$, and $Y_3$ are Hamiltonian vector fields with respect to $\omega_{Z_P}$. Moreover, since
\begin{eqnarray*}\{ {Z_P}h_1, {Z_P}h_2\}_{\omega_{Z_P}}&=& Z_Ph_1\,,\\
\{ {Z_P}h_2, {Z_P}h_3\}_{\omega_{Z_P}}&=& {Z_P}h_3\,,\\
\{ {Z_P}h_1,{Z_P}h_3\}_{\omega_{Z_P}}&=&2{Z_P}h_2\,,
\end{eqnarray*}
we see that $\overline{h}_1=Z_P h_1$, $\overline{h}_2=Z_Ph_2,$ and $\overline{h}_3=Z_Ph_3$ span a new finite-dimensional real Lie algebra. So,
if $h$ is a Hamiltonian on $\mathcal{O}_2$, then $\overline{h}=Z_Ph$ is another Hamiltonian on $\mathcal{O}_2$ taking the equivalent expression. 
So, 
\begin{equation}\label{barh}
 \overline{h}=\overline{h}_3+b_1(t)\overline{h}_1.
\end{equation}

To obtain a Nambu--Poisson structure, we need to find a volume form $\overline{\Omega}$ compatible with our structure $(\mathcal{O}_2,h,\bar{h})$.
It can be computed that in order to retrieve the initial first-order system \eqref{firstKS3} by means of the Nambu--Poisson brackets
\begin{equation}
\dot{x}=\{h,\overline{h},x\},\quad \dot{v}=\{h,\overline{h},v\},\quad \dot{a}=\{h,\overline{h},a\}
\end{equation}
 (recall: with $h,\overline{h}$, associated with the presymplectic forms $\omega_{3KS}$ and $\mathcal{L}_{Z_P}\omega_{3KS}$,)
the canonical volume form $\Omega=dx\wedge dv\wedge da$ has to be conformally transformed into
\begin{equation}\label{volformKS}
 \overline{\Omega}=\frac{i}{v^6}\left(\frac{a^2}{v^2}+4b_1(t)\right)^{3/2}dx\wedge dv\wedge da
\end{equation}
 \noindent
For this particular case, the canonical bracket $\{x,v,a\}=1$ turns out in
\begin{equation}\label{volformKSbrac}
 \{x,v,a\}=\frac{1}{\frac{i}{v^6}\left(\frac{a^2}{v^2}+4b_1(t)\right)^{3/2}}.
\end{equation}
Both expressions are real for values of the function $b_1(t)<-\frac{v^2}{4a^2}$.

It is also possible to find a third compatible presymplectic form by the action of the symmetry vector field on $\mathcal{L}_{Z_P}\omega_{3KS}$
It reads:
\begin{equation}\label{2ndpre}
 \bar{\bar{\omega}}=\mathcal{L}_{Z_P}\omega_{Z_P}=-\frac{4x}{v^2}dx\wedge da+\frac{2x^2}{v^3}dv\wedge da+\left(\frac{4xa}{v^3}+\frac{4}{v}\right)dx\wedge dv
\end{equation}
with corresponding Hamiltonian functions
\begin{align*}
 \iota_{Y_1}\bar{\bar \omega}&=d\left(\frac{4x^2}{v}\right),\quad  \iota_{Y_2}\bar{\bar \omega}=d\left(\frac{2x^2a}{v^2}-4x\right),
  &\iota_{Y_3}\bar{\bar \omega}=d\left(-\frac{4xa}{v}+\frac{x^2a^2}{v^3}+4v\right)
\end{align*}
when $c_0=0$.
\noindent
This presymplectic form \eqref{2ndpre} could have been equivalently used for the explained procedure above. 

If we apply the Hamilton--Jacobi theory to this problem, we have the following setting
\[
\xymatrix{ E=\mathcal{O}_2
\ar[dd]^{\pi} \ar[rrr]^{X_{h^{3KS}\bar{h}}}&   & &TE\ar[dd]^{T\pi}\\
  &  & &\\
N \ar@/^2pc/[uu]^{\gamma}\ar[rrr]^{X_{h^{3KS}\bar{h}}^{\gamma}}&  & & TN}
\]
where $\mathcal{O}_2=\{(x,v,a)\in T^{2}\mathbb{R}| v\neq 0\}$ and $N=\{(x,v)\in T\mathbb{R}| v\neq 0\}$.
The section $\gamma$ is locally given by $\gamma(x,v)=\gamma(x,v,\gamma^{a}(x,v)).$

The vector field $X_{h^{3KS}\bar{h}}$ can be obtained by performing the calculation
\begin{equation}\label{constr}
 X_{h^{3KS}\bar{h}}=\sharp (dh^{3KS}\wedge d\bar{h})
\end{equation}
Let refer us to their total derivatives with the following notation
\begin{equation}
 dh^{3KS}=h^{3KS}_xdx+h^{3KS}_vdv+h^{3KS}_ada
\end{equation}
where 
\begin{equation*}
h^{3KS}_x=\frac{\partial h^{3KS}}{\partial x}, \quad h^{3KS}_v=\frac{\partial h^{3KS}}{\partial v},\quad h^{3KS}_a=\frac{\partial h^{3KS}}{\partial a}
\end{equation*}
And for
\begin{equation}
 d\bar{h}=\bar{h}_xdx+\bar{h}_vdv+\bar{h}_ada
\end{equation}
where
\begin{equation*}
\bar{h}_x=\frac{\partial \bar{h}}{\partial x}, \quad \bar{h}_v=\frac{\partial \bar{h}}{\partial v},\quad \bar{h}_a=\frac{\partial \bar{h}}{\partial a}.
\end{equation*}
In this way, if we construct the vector field \eqref{constr}
{\begin{footnotesize}
\begin{equation}\label{cX}
 X_{h^{3KS}\bar{h}}=\left(\frac{\partial h^{3KS}}{\partial v}\frac{\partial \bar{h}}{\partial a}-\frac{\partial h^{3KS}}{\partial a}\frac{\partial \bar{h}}{\partial v}\right)\frac{\partial}{\partial x}+\frac{\partial h^{3KS}}{\partial a}\frac{\partial \bar{h}}{\partial x}\frac{\partial}{\partial v}
-\frac{\partial h^{3KS}}{\partial v}\frac{\partial \bar{h}}{\partial x}\frac{\partial}{\partial a}
\end{equation}
\end{footnotesize}}
and
{\begin{footnotesize}
\begin{equation}\label{pcX}
 \text{T}_{\gamma}X_{h^{3KS}\bar{h}}^{\gamma}=\left(\frac{\partial h^{3KS}}{\partial v}\frac{\partial \bar{h}}{\partial a}-\frac{\partial h^{3KS}}{\partial a}\frac{\partial \bar{h}}{\partial v}\right)\left(\frac{\partial }{\partial x}+\frac{\partial \gamma}{\partial x}\frac{\partial}{\partial a}\right)
+\frac{\partial h^{3KS}}{\partial a}\frac{\partial \bar{h}}{\partial x}\left(\frac{\partial}{\partial v}+\frac{\partial \gamma}{\partial v}\frac{\partial}{\partial a}\right)
\end{equation}
\end{footnotesize}}

Henceforth, denote by $\gamma_x=\frac{\partial \gamma}{\partial x}$ and $\gamma_v=\frac{\partial \gamma}{\partial v}$. By direct comparison of \eqref{cX} and \eqref{pcX},
according to \eqref{defvf}, we obtain the following equation
\begin{equation}
\left(\frac{\partial h^{3KS}}{\partial v}\frac{\partial \bar{h}}{\partial a}-\frac{\partial h^{3KS}}{\partial a}\frac{\partial \bar{h}}{\partial v}\right)\gamma_x+ \frac{\partial h^{3KS}}{\partial a}\frac{\partial \bar{h}}{\partial x}\gamma_v+\frac{\partial h^{3KS}}{\partial v}\frac{\partial \bar{h}}{\partial x}=0
\end{equation}
which corresponds with the Hamilton--Jacobi equation \eqref{echamj}.
If we try to solve this equation, we can find a particular solution
\begin{equation}
 \gamma^{a}=2a-\int{\frac{d\bar{h}}{\frac{\partial \bar{h}}{\partial a}}}.
\end{equation}

\subsection{N-coupled first-order Riccati equations}
The first-order Riccati equations can appear as a system of equations containing $n$ copies of same Ricatti equation \cite{Sardon}.
This system is
\begin{equation}\label{nric}
\dot{x}^{i_k}=a_0(t)+a_1(t)x^{i_k}+a_2(t){x^{i_k}}^2,\qquad i_k=1,\dots,n
\end{equation}
defined on $\mathcal{O}=\{(x_1,\dots,x_n)| (x_1-x_2) \dots  (x_{n-1}-x_n)\neq 0 \subset \mathbb{R}^n\}.$

The associated $t$-dependent vector field with this system is 
$$X_t=\sum_{i_k=1}^n \left(a_0(t)+a_1(t)x^{i_k}+a_2(t){x^{i_k}}^2\right)\frac{\partial}{\partial x^{i_k}}.$$

A method to obtain $n-1$ presymplectic forms from which to derive Hamiltonians, is the permutation of indices. By fixing one
of the coordinates, let us say $l$, with $l=1,\dots,n-1$
\begin{equation}
 \omega^{[l]}=\sum_{k<l}^{l-1}\frac{dx^k\wedge dx^{l}}{(x^k-x^{l})^2}+\sum_{k>l}^n\frac{dx^{l}\wedge dx^{k}}{(x^{l}-x^{k})^2},
\end{equation}
for a fixed $l$.

Due to its own construction, these are closed forms $d\omega^{[l]}=0$.
Equivalently, we can derive Hamiltonian functions associated with each $\omega^{[l]}$
\begin{multline*}
 h^{[l]}=a_0(t)\left(\sum_{k>l}^n\frac{1}{x^{l}-x^k}+\sum_{k<l}^{l-1}\frac{1}{x^k-x^{l}}\right)\\
+\frac{a_1(t)}{2}\left(\sum_{k>l}^n\frac{x^{l}+x^k}{x^{l}-x^k}+\sum_{k<l}^{l-1}\frac{x^{l}+x^k}{x^k-x^{l}}\right)\\
+a_2(t)\left(\sum_{k>l}^n\frac{x^{l}x^k}{x^{l}-x^k}+\sum_{k<l}^{l-1}\frac{x^{l}x^k}{x^k-x^{l}}\right)
\end{multline*}
for every fixed $l=1,\dots,n-1$.
According to the Nambu--Poisson theory, equations \eqref{nric} must be retrived through the computation
\begin{equation}\label{jacn}
 \dot{x}^{i_k}=\{h^{[1]},\dots,h^{[n-1]},x^{i_k}\},\quad i_k=1,\dots,n.
\end{equation}
and the $n$-dimensional bracket \eqref{defjacobian} takes the form
\begin{equation*}
\dot{x}^{i_k}=\{h^{[1]},h^{[2]},\dots,h^{[n-1]},x^{i_k}\}=
\left|
\begin{array}{cccc}
\frac{\partial h^{[1]}}{\partial x^1} & \frac{\partial h^{[1]}}{\partial x^2}& \dots & \frac{\partial h^{[1]}}{\partial x^n}\\
\dots & \dots & \dots & \dots\\
\frac{\partial h^{[k]}}{\partial x^1} & \frac{\partial h^{[k]}}{\partial x^2}& \dots & \frac{\partial h^{[k]}}{\partial x^n}\\
\dots & \dots & \dots & \dots\\
\frac{\partial h^{[n-1]}}{\partial x^{1}} & \frac{\partial h^{[n-1]}}{\partial x^2}& \dots & \frac{\partial h^{[n-1]}}{\partial x^n}\\
\frac{\partial x^{i_k}}{\partial x^1} & \frac{\partial x^{i_k}}{\partial x^2}& \dots & \frac{\partial x^{i_k}}{\partial x^n}\\
\end{array}
\right|
\end{equation*}
such that if we compute the determinant, we obtain
\begin{equation}\label{nricdet}
 \dot{x}^{i_k}=(-1)^{i_k+n}\sum_{\sigma_{i_1,\dots,i_n}} (-1)^{\left(\frac{n(n-1)}{2}+i_1+\dots+i_n\right)}\frac{\partial h^{[1]}}{\partial x^{i_1}}\dots \frac{\partial h^{[n-1]}}{\partial x^{i_n}},
\end{equation}
when $i_1,\dots,i_{n-1}=1,\dots,n\neq i_k$ and a particular $i_k$ that takes any value $1,\dots,n.$

The factor $\frac{\partial h^{[l]}}{\partial x^{j}}$ in \eqref{nricdet} takes the form

{\begin{footnotesize}
\begin{align*}
& \frac{\partial h^{[l]}}{\partial x^{j}}=a_0(t)\left(\delta_{j}^{[l]}\left(\sum_{l>k}^{l-1}\frac{1}{(x^k-x^l)^2}-\sum_{k>l}^n\frac{1}{(x^l-x^k)^2}\right)+{\bar{\delta}}_{j}^{[l]}\left(\sum_{k>l}^n\frac{1}{(x^{l}-x^k)^2}-\sum_{k<l}^{l-1}\frac{1}{(x^k-x^{l})^2}\right)\right)\\
+&a_1(t)\left(\delta_{j}^{[l]}\left(\sum_{l>k}^{l-1}\frac{x^k}{(x_k-x_l)^2}-\sum_{k>l}^n\frac{x^k}{(x_l-x_k)^2}\right)+{\bar{\delta}}_{j}^{[l]}\left(\sum_{k>l}^n\frac{x^l}{(x^{l}-x^k)^2}-\sum_{k<l}^{l-1}\frac{x^l}{(x^k-x^{l})^2}\right)\right)\\
+&a_2(t)\left(\delta_{j}^{[l]}\left(\sum_{l>k}^{l-1}\frac{(x^k)^2}{(x^k-x^l)^2}-\sum_{k>l}^n\frac{(x^k)^2}{(x^l-x^k)^2}\right)+{\bar{\delta}}_{j}^{[l]}\left(\sum_{k>l}^n\frac{(x^l)^2}{(x_{l}-x^k)^2}-\sum_{k<l}^{l-1}\frac{(x^l)^2}{(x^k-x^{l})^2}\right)\right)\\
\end{align*}
\end{footnotesize}}
where $l=1,\dots,n-1$ and $j=1,\dots,n.$

By comparison between \eqref{nric} and \eqref{nricdet}, we need to conformally transform the canonical volume form $\Omega=dx^1\wedge \dots \wedge dx^n$
associated with the former problem \eqref{nric} into another volume form $\bar{\Omega}$ corresponding with \eqref{nricdet}. 

There exists a change of coordinates $\overline{x}^{\hat j}=f_{\hat j}(x^j)x^j$ for all $\hat{j}=1,\dots,n$ through which we derive a compatible
$\overline{\Omega}$ compatible with \eqref{nricdet} that maps \eqref{nricdet} into \eqref{nric}. It takes the following form
\begin{equation}
 \overline{\Omega}=\prod_{\hat{j}=1,\dots,n} f_{\hat j}dx^1\wedge \dots \wedge dx^n
\end{equation}
with 
{\begin{footnotesize}
\begin{equation*}
 f_{\hat j}=\delta_{\hat j}^{[l]}\left(\sum_{l>k}^{l-1}\frac{1}{(x^k-x^l)^2}-\sum_{k>l}^n\frac{1}{(x^l-x^k)^2}\right)+{\bar{\delta}}_{\hat j}^{[l]}\left(\sum_{k>l}^n\frac{1}{(x^{l}-x^k)^2}-\sum_{k<l}^{l-1}\frac{1}{(x^k-x^{l})^2}\right)
\end{equation*}
\end{footnotesize}}
Therefore, the canonical $n$-dimensional Nambu--Poisson bracket takes the expression
\begin{equation}
 \{x^1,\dots,x^n\}=\frac{1}{\prod_{\hat{j}=1,\dots,n} f_{\hat j}}.
\end{equation}

If we want to apply the Hamilton--Jacobi theory to this example, we have the diagram
\[
\xymatrix{ E=\mathcal{O}
\ar[dd]^{\pi} \ar[rrr]^{X_{h^{1}\dots h^{n-1}}}&   & &TE\ar[dd]^{T\pi}\\
  &  & &\\
N \ar@/^2pc/[uu]^{\gamma}\ar[rrr]^{X_{h^{1}\dots h^{n-1}}^{\gamma}}&  & & TN}
\]
where $\mathcal{O}=\{(x_1,\dots,x_n)| (x_1-x_2) \dots  (x_{n-1}-x_n)\neq 0 \subset \mathbb{R}^n\}$ and $N=\{(x_1,\dots,x_n)| (x_1-x_2) \dots  (x_{n-2}-x_{n-1})\neq 0 \subset \mathbb{R}^{n-1}\}$.

The vector field $X_{h^{1}\dots h^{n-1}}$ can be obtained by performing the calculation
\begin{equation}
 X_{h^{1}\dots h^{n-1}}=\sharp (dh^{1}\wedge\dots \wedge dh^{n-1})
\end{equation}
In this way, it takes the expression
{\begin{footnotesize}
\begin{align*}
 X_{h^{1}\dots h^{n-1}}=&\sharp \left(\sum_{\sigma_{i_1\dots i_n}}(-1)^{\frac{n(n-1)}{2}+1_1+\dots+i_n}\frac{\partial h^{[1]}}{\partial x^{i_1}}\dots \frac{\partial h^{[n-1]}}{\partial x^{i_{n}}}\frac{dx^{i_1}\wedge \dots \wedge dx^{i_n}}{\left(\prod f^{i_1\dots i_n}\right)_{i_1,\dots,i_n}}\right)\\
&\sum_{\sigma_{i_1\dots i_n}}(-1)^{\frac{n(n-1)}{2}+1_1+\dots+i_n}\frac{\partial h^{[1]}}{\partial x^{i_1}}\dots \frac{\partial h^{[n-1]}}{\partial x^{i_{n}}}\frac{\delta^{i_1}_1 \dots \delta^{i_n}_{n-1}}{\left(\prod f^{i_1\dots i_n}\right)_{i_1,\dots,i_n}}\frac{\partial}{\partial x^{i_k}}
\end{align*}

\end{footnotesize}}
where $i_1,\dots,i_n\neq i_k=1,\dots,n.$

On the other hand, the vector field
{\begin{footnotesize}
\begin{align*}
\text{T}_{\gamma}& \left(X_{h^{1}\dots h^{n-1}}^{\gamma}\right)=\sharp \left(\sum_{\sigma_{i_1\dots i_n}}(-1)^{\frac{n(n-1)}{2}+1_1+\dots+i_n}\frac{\partial h^{[1]}}{\partial x^{i_1}}\dots \frac{\partial h^{[n-1]}}{\partial x^{i_{n}}}\frac{dx^{i_1}\wedge \dots \wedge dx^{i_n}}{\left(\prod f^{i_1\dots i_n}\right)_{i_1,\dots,i_n}}\right)\\
&\sum_{\sigma_{i_1\dots i_n}}(-1)^{\frac{n(n-1)}{2}+1_1+\dots+i_n}\frac{\partial h^{[1]}}{\partial x^{i_1}}\dots \frac{\partial h^{[n-1]}}{\partial x^{i_{n}}}\frac{\delta^{i_1}_1 \dots \delta^{i_n}_{n-1}}{\left(\prod f^{i_1\dots i_n}\right)_{i_1,\dots,i_n}}\left(\frac{\partial}{\partial x^{i_k}}+\frac{\partial \gamma^n}{\partial x^{i_k}}\frac{\partial}{\partial x^n}\right)
\end{align*}

\end{footnotesize}}
where $i_1,\dots,i_n\neq i_k$ and $i_k=1,\dots,n-1$ whilst $i_1,\dots,i_n=1,\dots,n.$ And we have chosen $\gamma$ in such a way as
$\gamma(x^1,\dots,x^{n-1},\gamma^{n}(x^1,\dots,x^{n-1}))$.

So, the Hamilton--Jacobi equation for this case reads
\begin{align*}
 \sum_{i_k}^{n-1} &\sum_{\sigma_{i_1\dots i_n}}^{1\leq i_j\leq n-1}(-1)^{\frac{n(n-1)}{2}+i_1+\dots+i_n}\frac{\partial h^{[1]}}{\partial x^{i_1}}\dots \frac{\partial h^{[n-1]}}{\partial x^{i_{n}}}\frac{\delta^{i_1}_1 \dots \delta^{i_n}_{n-1}}{\left(\prod f^{i_1\dots i_n}\right)_{i_1,\dots,i_n}}\gamma^n_{i_k}\\
&+\sum_{\sigma_{i_1\dots i_n}}^{1\leq i_j\leq n}(-1)^{\frac{n(n-1)}{2}+i_1+\dots+i_n}\frac{\partial h^{[1]}}{\partial x^{i_1}}\dots \frac{\partial h^{[n-1]}}{\partial x^{i_{n}}}\frac{\delta^{i_1}_1 \dots \delta^{i_n}_{n-1}}{\left(\prod f^{i_1\dots i_n}\right)_{i_1,\dots,i_n}}=0,
\end{align*}
where $\gamma^n_{i_k}$ means $\frac{\partial \gamma^n}{\partial x_{i_k}}.$

\section{Conclusions}
In this paper we have developed a Hamilton--Jacobi theory for Nambu--poisson systems, extending the classical approach for Hamiltonian systems.
We apply the theoretical results to two examples. One is the third-order Kummer--Schwarz equation and the other is the $n$-coupled first-order
Riccati equations.
To do this, we previously study the properties that characterize the Lagrangian submanifolds for Nambu--structures.

\section*{Acknowledgements}
This work has been partially supported by MINECO MTM 2013-42-870-P and
the ICMAT Severo Ochoa project SEV-2011-0087.

\end{document}